\documentclass[letterpaper, 10pt, final, conference, twocolumn]{ieeeconf}

\IEEEoverridecommandlockouts

\usepackage{epsfig} %

\usepackage{amsmath,graphicx,amsfonts,amssymb,epsfig,subfig,mathrsfs,mathtools}

\usepackage{ntheorem}
\usepackage{cuted}
\usepackage{arydshln}
\usepackage{blkarray}

\usepackage{color}
\usepackage{multirow}
\usepackage{multicol}
\usepackage{lipsum}
\usepackage{rotating}
\usepackage{graphicx}
\usepackage{psfrag}
\usepackage{algorithm}

\usepackage[utf8]{inputenc}
\usepackage{xspace}

\usepackage{booktabs}
\usepackage{algorithmicx}
\usepackage{algpseudocode}
\usepackage{cite}
\usepackage{tikz}
\usepackage{cancel}
\usepackage{textcomp}
 \definecolor{lin}{RGB}{240,0,0}
 \definecolor{paleblue}{RGB}{0,9,255}

\usepackage{comment}
\includecomment{comment}
\specialcomment{comment}{\begingroup}{\endgroup}

\newcommand{\map}[3]{#1: #2 \rightarrow #3}
\newcommand{\setdef}[2]{\{#1 \; | \; #2\}}

\newcommand{\real}{\ensuremath{\mathbb{R}}}
\newcommand{\probP}{\ensuremath{\mathbb{P}}}
\newcommand{\probQ}{\ensuremath{\mathbb{Q}}}
\newcommand{\probW}{\ensuremath{\mathbb{W}}}

\newcommand{\unitcircle}{\ensuremath{\mathbb{S}^1}}
\newcommand{\realnonnegative}{\ensuremath{\mathbb{R}}_{\ge 0}}
\newcommand{\integernonnegative}{\ensuremath{\mathbb{Z}}_{\ge 0}}
\newcommand{\integerpositive}{\ensuremath{\mathbb{Z}}_{> 0}}
\newcommand{\until}[1]{\{1,\dots, #1\}}
\newcommand{\subscr}[2]{#1_{\textup{#2}}}
\newcommand{\supscr}[2]{#1^{\textup{#2}}}
\newcommand{\vect}[1]{\boldsymbol{#1}}
\newcommand{\vectorones}[1]{\vect{1}_{#1}}

\newcommand{\Norm}[1]{\|#1\|}
\newcommand{\trans}[1]{{#1}^\top}

\newcommand{\Prob}{\ensuremath{\operatorname{Prob}}}

\newcommand{\untilinterval}[2]{\{{#1},\dots, {#2}\}}

\newcommand{\pinverse}[1]{{#1}^{\dagger}}

\newcommand{\qed}{~\hfill{ $\blacksquare$}}

\graphicspath{{./Fig/}}
\usepackage{epstopdf}
\makeatletter
\newtheoremstyle{breaknote}
  {\item{\theorem@headerfont
          ##1\ ##2\theorem@separator}\hskip\labelsep\relax}
  {\item{\theorem@headerfont
          ##1\ ##2\ (##3)\theorem@separator}\hskip\labelsep\relax}
\makeatother
\theoremstyle{breaknote}
\newtheorem{assumption}{Assumption}[section]
\newtheorem{theorem}{Theorem}[section]

\newtheorem{lemma}{Lemma}[section]

\newtheorem{remark}{Remark}[section]

\newtheorem{example}{Example}[section]

\renewcommand{\baselinestretch}{.97}

\title{Online Learning of Parameterized Uncertain Dynamical
  Environments \\ with Finite-sample Guarantees}

\author{Dan Li$^{1}$, Dariush Fooladivanda$^{2}$ and Sonia
  Mart{\'\i}nez$^{1}$ \thanks{*This research was developed with
    funding from ONR N00014-19-1-2471, and AFOSR FA9550-19-1-0235.}
  \thanks{$^{1}$ D. Li and S. Mart{\'\i}nez are with the Department of
    Mechanical and Aerospace Engineering, University of California San
    Diego, La Jolla, CA 92092, USA. {\tt lidan@ucsd.edu;
      soniamd@ucsd.edu} }
\thanks{$^{2}$ D. Fooladivanda is with the
    Department of Electrical Engineering and Computer Sciences,
    University of California at Berkeley, Berkeley, CA 94720,
    USA. {\tt  dfooladi@berkeley.edu;}}
}

\begin{document}

\maketitle
\begin{abstract}
  We present a novel online learning algorithm for a class of unknown
  and uncertain dynamical environments that are fully
  observable. First, we obtain a novel probabilistic characterization
  of systems whose mean behavior is known but which are subject to
  additive, unknown subGaussian disturbances. This characterization
  relies on recent concentration of measure results and is given in
  terms of ambiguity sets.
  Second, we extend the results to environments whose mean behavior is
  also unknown but described by a parameterized class of possible mean
  behaviors. Our algorithm adapts the ambiguity set
  dynamically by learning the parametric dependence online, and
  retaining similar probabilistic guarantees with respect to the
  additive, unknown disturbance. We illustrate the results on a
  differential-drive robot subject to environmental uncertainty.
\end{abstract}

\section{Introduction}
The online learning of uncertain dynamical systems has broad
application in various domains, including those of artificial
intelligence and
robotics~\cite{AC-GP:19,AHQ-YM-AS-MCY:20,AS-RM-JM-AJI:08}.
Fundamentally,
one is to exploit input-output data to identify the representation of
the environment that best captures its
behavior.
In this way, several techniques, from first-principles system
identification to, more recently, (deep) neural networks, have been
successfully used in various domains. Unfortunately, safe performance
usually depends upon the assimilation of vast amounts of data, which
is mostly done offline and prevents its  application in real-time
scenarios.
Motivated by this, we investigate the integration of
recently-developed probabilistically-guaranteed system descriptions
with online, predictor-based learning algorithms.

The system identification literature broadly encompasses
linear~\cite{LL:99,MV-VV:07} and non-linear
systems~\cite{MM-CN:11,CN-AN-GCC:19}, with asymptotic performance
guarantees. More recently, finite-sample analysis of identification
methods have been proposed for linear systems~\cite{TS-AR:19,SO-NO:19,
  AT-GJP:19,SF-NM-SS:19}. These methods leverage modern
measure-of-concentration results~\cite{NF-AG:15,RV:18} for
non-asymptotic guarantees of the identification error
bounds. Measure-of-concentration results are also used
in~\cite{DB-JC-SM:19-ecc,DB-JC-SM:19-tac}. However, the goal
of~\cite{DB-JC-SM:19-ecc,DB-JC-SM:19-tac} is to learn an unknown
initial distribution evolving under a known dynamical system
while assimilating data via a linear observer.  This characterization
is given in terms of \textit{ambiguity sets}, which are constructed
via multiple system trajectories or realizations.
In contrast, here we employ Wasserstein metrics to develop an online
learning algorithm for uncertain dynamical systems with
similar-in-spirit probabilistic guarantees.

\noindent \textit{Statement of Contributions:} We propose an online
learning algorithm that characterizes a class of unknown and uncertain
dynamical environments with probabilistic guarantees using a finite
amount of data. To achieve this, we first assume that the mean
behavior of the stochastic system is known but the system states are
subject to an additive, unknown subGaussian distribution,
characterized by a set of distributions or \textit{ambiguity set}.
Then, we extend the results to environments whose mean behavior is
unknown but belongs to a parameterized class of behaviors. In this
regard, we propose a time-varying parameterized ambiguity set and a
learning methodology to capture the behavior of the environment. We
show how the proposed online learning algorithm retains desirable
probabilistic guarantees with high confidence. A differential-drive
robot subject to environmental uncertainty is provided for an
illustration. Basic notations and definitions can be found in the
footnote.
\footnote{Let $\real^m$,
  $\realnonnegative^m$, $\integernonnegative^m$ and $\real^{m \times
    n}$ denote respectively the $m$-dimensional real space, the
  $m$-dimensional nonnegative real space, the $m$-dimensional
  nonnegative integer space, and the space of $m \times n$
  matrices. By $\vect{x} \in \real^m$ we denote a column vector of
  dimension $m$, while $\trans{\vect{x}}$ represents its
  transpose. The shorthand notation $\vectorones{m}$ denotes the
  column vector $\trans{(1,\cdots,1)} \in \real^m$. We use subscripts
  to index vectors, i.e., $\vect{x}_k \in \real^m$ for $k \in
  \integernonnegative$, and we use $x_i$ to denote the
  $\supscr{i}{th}$ component of $\vect{x}$. We denote respectively the $2$-norm and $\infty$-norm by $\Norm{\vect{x}}$ and $\Norm{\vect{x}}_{\infty}$.
  We define the $m$-dimensional norm ball with center $\vect{x} \in
  \real^m$ and radius $\epsilon \in \realnonnegative$ as the set
  $B_{\epsilon}(\vect{x}):=\setdef{\vect{y} \in
    \real^m}{\Norm{\vect{y}-\vect{x}} \leq \epsilon}$.  We denote by
  $\langle \cdot, \cdot \rangle$ an inner product in the space of
  interest. Consider the space $\real^m$, we define $\langle \vect{x},
  \vect{y} \rangle:= \trans{\vect{x}} \vect{y}$, $\vect{x}, \vect{y}
  \in \real^m$. In particular, $\Norm{\vect{x}}:=\sqrt{\langle
    \vect{x}, \vect{x} \rangle}$. Consider Finsler manifold $\real^2
  \times [ -\pi, \pi) \cong \real \times \unitcircle$ where
  $\unitcircle$ stands for the unit circle. For $(\vect{x},\theta_1),
  (\vect{y},\theta_2) \in \real^2 \times [ -\pi, \pi)$, we define
  $\langle (\vect{x},\theta_1), (\vect{y},\theta_2) \rangle:=
  \trans{\vect{x}} \vect{y} + \cos( \min \{ |{\theta_1-\theta_2}|, \;
  2 \pi -|{\theta_1-\theta_2}| \})$. In particular, we use $
  \Norm{(\vect{x},\theta)}:= \sqrt{\langle \vect{x}, \vect{x} \rangle
    + 1}$.
  Given an $A \in \real^{m \times m}$, we write its Singular
  Value Decomposition (SVD) as $A=U \Sigma \trans{V}$, where $U$, $V
  \in \real^m$ are orthonormal and $\Sigma$ is diagonal with non-negative entries.
  These entries are called singular values of $A$, and we denoted by $\subscr{\sigma}{max}(A)$ and
  $\subscr{\sigma}{min}(A)$ the maximal and non-zero minimal singular
  value of $A$, respectively. We denote by $\pinverse{A}:=V
  \pinverse{\Sigma} \trans{U}$ the Moore–Penrose inverse of $A$, where
  $\pinverse{\Sigma}$ is the same as $\Sigma$ except the replacement of each positive entry by its inverse.
  Let
  $(\Omega,\mathcal{F},\probP)$ be a probability space, with $\Omega$
  the sample space, $\mathcal{F}$ a $\sigma$-algebra, and $\probP$ the
  associated probability distribution. Let
  $\map{\vect{x}}{\Omega}{\real^m}$ be an induced random vector.  We
  denote by $\mathcal{M}$ the space of all probability distributions
  with finite first moment. %
  To measure the distance in $\mathcal{M}$, we use the
  dual version of the $1$-Wasserstein metric %
  $\map{d_W}{\mathcal{M} \times \mathcal{M}}{\realnonnegative}$,
  defined
  as in~\cite{KLV-RGS:58}. %
 A closed Wasserstein ball of
 radius $\epsilon$ centered at a distribution $\probP \in
 \mathcal{M}$ is denoted by
 $\mathbb{B}_{\epsilon}(\probP):=\setdef{\probQ \in
   \mathcal{M}}{d_W(\probP,\probQ) \le
   \epsilon}$. We denote the Dirac measure at $x_0 \in \Omega$
   as $\map{\delta_{\{x_0\}}}{\Omega}{\{0,1 \}}$. For any set $A \in
   \mathcal{F}$, we let $\delta_{\{x_0\}}(A)=1$, if $x_0 \in A$,
   otherwise $0$. For an $\vect{x}\in \Omega$, we denote $\probP \equiv \probQ + \vect{x}$, if $\probP$ is a translation of $\probQ$ by $\vect{x}$.

}

\section{Problem Statement}\label{sec:ProbStat}

This section presents the description of the uncertain dynamical
environment which we aim to learn, with a problem
definition.
Let $t\in
\integernonnegative$ denote time discretization. For each $t$, the
uncertain system is characterized by a random variable $\vect{x} \in
\real^n$ which evolves according to an \textit{unknown},
discrete-time, stochastic and, potentially, time-varying system
\begin{equation}
 \begin{aligned}
   \vect{x}_{t+1}=&f(t,\vect{x}_{t},\vect{d}_{t}) + \vect{w}_{t}, \;
   {\textrm{ with some }}
   \vect{x}_0 \sim \probP_0. %
 \end{aligned} \label{eq:env}
\end{equation}
The distribution $\probP_{t+1}$ characterizing
$\vect{x}_{t+1}$ is determined by the current state's distribution,
the unknown mapping %
$\map{f}{\realnonnegative \times \real^n \times \real^m}{\real^n}$,
and random vectors $\vect{w}_{t}$ that cannot be captured by $f$.  We
further assume that $\vect{d}_t$ is an exogenous signal that is
selected in advance or revealed online, which can play the role of an
external reference or control. %
Let us denote by $\probW_t$ the distribution of the random vector
$\vect{w}_{t} \in \real^n$. %

\begin{assumption}[Independent and stationary subGaussian distributions] \label{assump:subG}
{\rm Consider random vectors $\vect{w}_{t} \in \real^n$, $t \in \integernonnegative$. It is assumed that:
  \noindent \textbf{(1)} The random vectors $\vect{w}_{t}$ are
  component-wise and time-wise independent, i.e., $w_{t,i}$ and
  $w_{k,j}$ are independent, for all $t \ne k$, $i \ne j$, $(t,k) \in
  \integernonnegative^2$ and $(i,j) \in \untilinterval{1}{n}$.
  \textbf{(2)} For each $t$, $\vect{w}_{t}$ is a %
   zero-mean
  $\sigma$-subGaussian, i.e., for any $a \in \real^n$ we have
  $\mathbb{E} \left[ \exp ( \trans{a} {\vect{w}_{t}}) \right] \leq
  \exp ( {\Norm{a}^2 \sigma^2}/{2})$.  }
\begin{example}[$\sigma$-subGaussian distributions]
  {\rm A trivial example is any $\probW \equiv
    \mathcal{N}(\vect{0},\Sigma)$ with $\subscr{\sigma}{max}(\Sigma)
    \leq \sigma^2$. As any random vector supported on a compact set
    belongs to the subGaussian class, in particular, the
    following are $\sigma$-subGaussian distributions: \textbf{(1)}
    any zero-mean uniform distribution $\vect{w} \sim
    \mathcal{U}(\Omega)$ supported over $\Omega \subset
    B_{\sigma}(\vect{0})$; \textbf{(2)} any zero-mean discrete
    distribution with support $\Omega \subset B_{\sigma}(\vect{0})$. }
\end{example}

\end{assumption}
This paper aims to obtain a tractable characterization of the unknown
distribution $\probP_{t+1}$ of the immediate-future environment state
$\vect{x}_{t+1}$ online, $\forall \, t$. This is to be done %
by employing historical measurements, $\hat{\vect{x}}_{k}$, ${k \leq
  t}$, and data $\hat{\vect{d}}_{k}$, ${k \leq t}$. %

\begin{remark}[On finite-horizon learning
  of~\eqref{eq:env}] \label{remark:any_gamma} {\rm Our learning
    problem can be extended over finite horizons as follows. Let $N$
    be a learning horizon, then for each $t$ the goal is to
    characterize the dynamical environment over the next $N$ time
    slots, $\untilinterval{t+1}{t+N}$, with the previous knowledge of
    $\vect{d}:=(\vect{d}^{(1)},\ldots, \vect{d}^{(N)})$.  In other
    words, the objective is to characterize the
  joint distribution $\probQ:=\probP_{t+1} \times \cdots
    \times \probP_{t+N}$ of the stochastic process
    $\vect{x}:=(\vect{x}^{(1)}, \ldots, \vect{x}^{(N)})$.
  }
\end{remark}

\section{Characterization of Random Dynamical
  Environments under Perfect Information}
We aim to provide a description the random dynamical
system~\eqref{eq:env} via ambiguity sets.  %
More precisely, given knowledge $\vect{d}$, and system data
$\hat{\vect{x}}$, we look for a set of distributions
$\mathcal{P}_{t+1}:=\mathcal{P}_{t+1}(\vect{d},\hat{\vect{x}})$
characterizing $\probP_{t+1}$ via %
\begin{equation}
  {\Prob}\left( \probP_{t+1} \in \mathcal{P}_{t+1} \right)\geq 1-
  \beta,
  \label{eq:guarantee}
\end{equation}
for some $\beta \in(0,1)$. Observe that the probability is taken wrt the
historical random data
outcomes. %
To do this, let $T_0\in \integerpositive$ and $T:=\min \{t, T_0 \} \ge
1$, and consider the
historical data, $\hat{\vect{x}}_{k}$ and $\hat{\vect{d}}_{k}$, for $k
\in \mathcal{T}%
:=\{t-T,\ldots,t-1 \}$.
Assuming a perfect knowledge of $f$, we show first how to use the data
set $\mathcal{I}:=\{ \hat{\vect{x}}_{t}, \; \hat{\vect{x}}_{k}, \;
\hat{\vect{d}}_{k}, \; k \in \mathcal{T}%
 \}$  to construct
$\mathcal{P}_{t+1}$, $\forall \, t \ge 0$.
Let us denote by $\probQ_{t+1}\equiv\probQ_{t+1}(\vect{d})$ the
empirical distribution of $\vect{x}_{t+1}$ and define it as follows
\begin{equation*}
  \probQ_{t+1} %
  :=\frac{1}{T} \sum\limits_{k\in \mathcal{T}} \delta_{\{ {\xi}_{k}(\vect{d}) \}},
\end{equation*}
where $ {\xi}_{k}(\vect{d}):= f({t},\hat{\vect{x}}_{t},
  \vect{d}) + \hat{\vect{x}}_{k+1}
  -f({k},\hat{\vect{x}}_{k},\hat{\vect{d}}_{k}),~\forall
  k\in\mathcal{T}.$
  The following result enables us to construct the ambiguity set
  $\mathcal{P}_{t+1}$ that satisfies (\ref{eq:guarantee}).
\begin{lemma}[Asymptotic dynamic ambiguity set]
  Let us assume that %
  the system $f$ is known at each time $t$.
  Given a confidence level $\beta \in (0,1)$, parameter $T_0 \in
  \integerpositive$, and horizon $T=\min \{t, T_0 \}$, let us assume $\vect{w}_{k}$ is i.d. for $k \in \mathcal{T}$. Then,
  there exists a
  positive scalar $\epsilon:=\epsilon(T,\beta)$ such
  that~\eqref{eq:guarantee} holds by selecting
\begin{equation*}
  \mathcal{P}_{t+1}%
  := \mathbb{B}_{\epsilon}(\probQ_{t+1}) %
  =\setdef{\probP}{ d_W(\probP, \probQ_{t+1}) \leq \epsilon},
\end{equation*}
a Wasserstein ball centered at $\probQ_{t+1}$ with radius
\begin{equation*}
  \epsilon:= \sqrt{ \frac{2n \sigma^2}{T} \ln(\frac{1}{\beta})} + \mathcal{O}( T^{- {1}/{\max \{ n, 2  \}  }} ),
\end{equation*}
where $n$ is the dimension of $\vect{x}$ and $\sigma$ is as in
Assumption~\ref{assump:subG}. Further, if $T_0=
\infty$, %
then as $t \rightarrow \infty$, $\epsilon \rightarrow 0$, i.e., the
set $\mathcal{P}_{t+1}$ shrinks to the singleton $\probP_{t+1}$ at a
rate $\mathcal{O}({1}/{ T^{- {1}/{\max \{ n, 2 \} } }})$.
\label{lemma:perfectenv}
\end{lemma}
In practice, $T_0$, and $\beta$ need to be selected empirically,
in order to efficiently address the particular problem that leverages
the characterization of~\eqref{eq:env}.
We provide all the proofs of the lemmas and theorems in Appendix.

\section{Characterization of Random Dynamical Environments in a
  Parameterized Family}
The construction of the empirical distribution $\probQ_{t+1}$ of the
previous section relies on the knowledge of $f$.
When $f$ is unknown, one may represent $f$ as belonging to a
parameterized class of functions. Such as the approach adopted in the
neural networks field and Koopman operator theory.
Here, we focus on the case that $f$ is approximated by a linear
combination of a class of functions or ``predictors'' as follows. %
\begin{assumption}[Environment predictor
  class] \label{assump:predictor} {\rm There exists a set of predictors
    $\map{f^{(i)}}{\realnonnegative \times \real^n \times
   \real^m}{\real^n}$, $(t,\vect{x},\vect{d}) \mapsto
 f^{(i)}(t,\vect{x},\vect{d})$, $i \in \until{p}$,  such that: (1) The vector fields
    $f^{(1)},f^{(2)}, \ldots,f^{(p)}$ are linearly independent almost everywhere.
    (2)
    There exists potentially time-varying coefficients
    $\vect{\alpha}^{\star}:=(\alpha^{\star}_1,\ldots,\alpha^{\star}_p)
    \in\real^p$ such that
\begin{equation*}
  f(t,\vect{x},\vect{d}) = \sum\limits_{i=1}^{p} \alpha_{i}^\star f^{(i)}(t,\vect{x},\vect{d}).
\end{equation*}
}%
\end{assumption}
As the selection of the predictors is not the subject of this study,
we assume that the predictors are found in advance, and hence they are
known to the learning algorithm.

The construction of an effective ambiguity set now depends on learning
the dynamical environment mapping. Let us denote by
$\vect{\alpha}\equiv\vect{\alpha}_t$ the estimated value of the
parameter $\vect{\alpha}^\star$ at time $t$. To construct
$\mathcal{P}_{t+1}$, %
consider $T$ predictions of $\vect{x}_{t+1}$ using
$f^{(i)}$, denoted by ${\xi}_{k}^{(i)}(\vect{\alpha},\vect{d})$. For
each $k \in \mathcal{T}$, $i\in\{1,\cdots,p\}$, and given
$\vect{d}:=\vect{d}_t$, we define
\begin{equation*} {\xi}_{k}^{(i)}(\vect{\alpha},\vect{d}):=
  f^{(i)}({t},\hat{\vect{x}}_{t}, \vect{d}) +
  \frac{\hat{\vect{x}}_{k+1}}{\trans{\vect{\alpha}} \vectorones{p}}
  -f^{(i)}({k},\hat{\vect{x}}_{k},\hat{\vect{d}}_{k}). %
\end{equation*}
Now, we select the empirical
$\hat{\probP}_{t+1}\equiv\hat{\probP}_{t+1}(\vect{\alpha},\vect{d})$,
as follows:
\begin{equation} \label{eq:empdistn}
  \hat{\probP}_{t+1} %
  := \frac{1}{T} \sum\limits_{k\in \mathcal{T}} \delta_{\{
    \sum\limits_{i=1}^{p} \alpha_i
    \xi_k^{(i)}(\vect{\alpha},\vect{d})\}}.
\end{equation}
The following result enables the construction of the ambiguity set
$\mathcal{P}_{t+1}$, relying on both $\vect{d}$ and $\vect{\alpha}$,
which satisfies (\ref{eq:guarantee}).
\begin{theorem}[Adaptive dynamic ambiguity set] \label{thm:ambiguityP}
  Assume
  that %
  the data set $\mathcal{I}$ is accessible, $\forall \, t$.  Further,
  let Assumption~\ref{assump:predictor}, on the environment predictor
  class, hold for some $\vect{\alpha}^{\star}$ %
  at time $t \in \mathcal{T}$. Then, given a confidence level $\beta \in (0,1)$, horizon
  parameter $T_0$, %
  and a learning parameter $\vect{\alpha}\equiv\vect{\alpha}_t \in
  \real^p$,
    there
    exists a scalar
    $\hat{\epsilon}:=\hat{\epsilon}(t,T,\beta,\vect{\alpha},\vect{d})$
    such that~\eqref{eq:guarantee} holds by selecting
\begin{equation*}
  \mathcal{P}_{t+1} %
  := \mathbb{B}_{\hat{\epsilon}%
  }(\hat{\probP}_{t+1} %
  ) =\setdef{\probP}{ d_W(\probP, \hat{\probP}_{t+1}) \leq \hat{\epsilon} },
\end{equation*}
where $
    \hat{\epsilon} %
=
\epsilon %
+ \Norm{\vect{\alpha}^{\star} -\vect{\alpha}}_{\infty} H(t,T,{\vect{d}}),
$
with
\begin{equation*}
  H(t,T,{\vect{d}}):= \frac{1}{T} \sum\limits_{i=1}^{p} \sum\limits_{k\in \mathcal{T}} \Norm{ f^{(i)}(k,\hat{\vect{x}}_{k},\hat{\vect{d}}_{k}) - f^{(i)}(t,\hat{\vect{x}}_{t},{\vect{d}}) },
\end{equation*}
and the radius $\epsilon$ is selected as in
Lemma~\ref{lemma:perfectenv}.
\end{theorem}
Theorem~\ref{thm:ambiguityP} indicates that, if we select
$\vect{\alpha}$ wisely, i.e., $\vect{\alpha} \equiv
\vect{\alpha}^{\star}$, then the adaptive dynamic ambiguity set is
identical to that of Lemma~\ref{lemma:perfectenv}. %

To estimate an unknown $\vect{\alpha}^{\star}$ while preserving the
probabilistic guarantees, we propose an online learning
algorithm %
that attempts to bring $\vect{\alpha}$ close to
$\vect{\alpha}^{\star}$ with high probability. %
Intuitively, our approach is based on the comparison of new obtained
data with updates given by a predictor combination.%
\begin{theorem}[Learning of $\vect{\alpha}^\star$] \label{thm:proj} Let the
  data set $\mathcal{I}$ and predictors $\{
  f^{(i)}\}_i$ %
  be given.
  For each $k \in \mathcal{T}$ and $i \in \until{p}$, let us denote
  $f^{(i)}_k:=f^{(i)}({k},\hat{\vect{x}}_{k},\hat{\vect{d}}_{k})$.
  Consider the data matrix $A \equiv A_t \in \real^{p \times p}$ with
\begin{equation*}
  \begin{aligned}
    A(i,j):= \frac{1}{T} \sum\limits_{k\in \mathcal{T}} \langle f^{(j)}_{k} ,  P_k  f^{(i)}_{k}  \rangle, \; \; i,\; j \in \until{p}, \\
    \end{aligned}
\end{equation*}
where $P_k$ is an online regularization matrix
at time $k$, and let us consider the data vector $\vect{b}\equiv
\vect{b}_t \in \real^{p}$,
with
components %
\begin{equation*}
  \begin{aligned}
    \vect{b}(i):= \frac{1}{T}\sum\limits_{k\in \mathcal{T}} \langle
    \hat{\vect{x}}_{k+1} , P_k f^{(i)}_{k} \rangle, \;\; i \in
    \until{p}.
    \end{aligned}
\end{equation*}
Given $\eta>0$, we select $P_k$ such that $\Norm{ P_k f^{(i)}_{k}}
\leq \eta$ for all $i \in \until{p}$, $k \in \mathcal{T}$, and select
$\vect{\alpha}\equiv\vect{\alpha}_{t}$ to be
\begin{equation}
  \vect{\alpha} = \pinverse{A} \vect{b},
\label{eq:alpha}
\end{equation}
where $\pinverse{A}$ denotes the Moore–Penrose inverse of $A$.  Let
Assumption~\ref{assump:subG} and Assumption~\ref{assump:predictor}
hold, and take
\begin{equation*}
  c:= \sigma e   \eta \sqrt{np} \subscr{\sigma}{min}^{-1} (A),
\end{equation*}
where $\sigma$ is that in Assumption~\ref{assump:subG}, the constant
$e\approx2.718$,
 and $\subscr{\sigma}{min} (A)$ is the
minimal non-zero principal singular value of $A$. Then by selecting
$\gamma \geq n c$, the parameter $\vect{\alpha}$ is ensured to be close
to $\vect{\alpha}^{\star}$ with high probability in the following
sense:
\begin{equation*}{\small
\Prob \left( \Norm{\vect{\alpha} - \vect{\alpha}^{\star}}_{\infty}
  \leq \gamma \right)
\geq
    1- \exp \left( - \frac{ (n c - \gamma)^2 T^2 }{2 \left[ (2T-1)c \gamma + n c^2 \right]} \right).
}
\end{equation*}
In particular, selecting $\gamma \geq nc/e$, we obtain a non-trivial
bound with a slow confidence growth rate as follows
\begin{equation*}
\Prob \left( \Norm{\vect{\alpha} -
        \vect{\alpha}^{\star}}_{\infty} \leq \gamma \right) \geq 1-
    \frac{1 }{\gamma} n \sigma \eta \sqrt{np}\subscr{\sigma}{min}^{-1} (A).
\end{equation*}
\end{theorem}
Theorem~\ref{thm:proj} provides an online computation of a real-time
$\vect{\alpha}$ that is close to %
$\vect{\alpha}^{\star}$ %
within a time varying distance $\gamma$ with arbitrary high
probability, where this distance $\gamma$ depends only on the
environment predictors as well as on the data sets. Note that, the
confidence of selecting $\gamma>nc$ as a bound of $\Norm{\vect{\alpha}
  - \vect{\alpha}^{\star}}_{\infty}$ increases exponentially as we
increase the length $T$ of the data sets.  This motivates us to
propose a computable dynamic ambiguity set, described as in
Theorem~\ref{thm:ambiguityP}, by selecting its dynamic radius as
\begin{equation}
  \begin{aligned}
\hat{\epsilon} %
=&
\epsilon %
+ \gamma H(t,T,{\vect{d}}),
  \end{aligned}
  \label{eq:tractable_epsilon}
\end{equation}
where $\epsilon$, $\gamma>nc$ and $H$ are chosen as in
Lemma~\ref{lemma:perfectenv}, Theorem~\ref{thm:proj}, and
Theorem~\ref{thm:ambiguityP}, respectively. Such selection results in
modified guarantees of~\eqref{eq:guarantee} as follows
\begin{equation}
  \begin{aligned}
    &
{\Prob}\left( \probP_{t+1} \in
      \mathcal{P}_{t+1} %
    \right)
\\
    & \hspace{2ex}
\geq \left( 1- \beta \right) \left( 1- \exp
      \left( - \frac{ (n c - \gamma)^2 T^2 }{2 \left[ (2T-1)c \gamma +
            n c^2 \right]} \right) \right),
 \end{aligned}
\label{eq:modgua}
\end{equation}
where as time $t$ increases with a selection of $T_0= \infty$ (or $T =
t$), the confidence value on the right hand side increases to
$1-\beta$ exponentially fast. Fig.~\ref{fig:adaptive_amb_set}
compares the adaptation of the ambiguity set with and without knowing
$f$.
\begin{figure}[tbp]%
\centering
\includegraphics[width=0.2\textwidth]{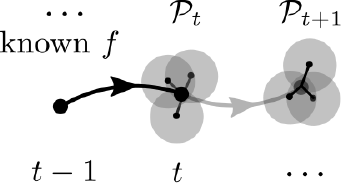}%
\hspace{6ex}
\includegraphics[width=0.22\textwidth]{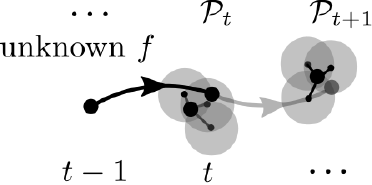}%
\caption{\small
Online characterization of $\mathcal{P}_{t+1}$, with (without) $f$. The
dark line is the trajectory of $\vect{x}$ and the gray part is yet to
be revealed. At $t$, we obtain $\mathcal{P}_{t+1}$ with its elements supported on $T_0=3$ shaded regions with high probability. Each region $k$ has center $\xi_k$ $(\sum_i \alpha_i \xi^{(i)}_k)$ and radius proportional to $\epsilon$ $(\hat{\epsilon})$. Note the centers of these regions are related to a known (learned) point $f(t,\hat{\vect{x}}_{t},\vect{d})$ $(\sum_i \alpha_i f^{(i)}(t,\hat{\vect{x}}_{t},\vect{d}))$,
they are close if the learning is effective.
}%
\label{fig:adaptive_amb_set}%
\end{figure}
\begin{remark}[Data-driven selection of the
  radius] \label{remark:select_eps} {\rm The radius of the adaptive
    ambiguity set~\eqref{eq:tractable_epsilon} depends on the unknown,
    noise-related parameter $\sigma$, the regularization constant
    $\eta$, and on the online parameters $\subscr{\sigma}{min}
    (A)$. In many engineering problems, an upper bound $\sigma$ of the
    noise-related parameter can be determined \textit{a-priori} or
    empirically. The parameter $\eta$, together with the
    regularization matrices $P$, are introduced to ensure
    that~\eqref{eq:alpha} is well posed. In particular, $P$ can be a
    diagonal matrix with each diagonal term scaling its corresponding
    components. At each $t$, the computation~\eqref{eq:alpha} needs an
    additional online regularization matrix, denoted by
    $P_{t-1}$.
    For example, $P_{t-1}$ can be a diagonal matrix with the
    $\supscr{j}{th}$ diagonal term equal to ${\small
      1/(\sqrt{p}\max_{i \in \until{p}} |f^{(i)}_{t-1}(j)|)}$, where
    $f^{(i)}_{t-1}(j)$ is the $\supscr{j}{th}$ component of
    $f^{(i)}_{t-1}$, which results in $\eta=1$. Finally,
    $\subscr{\sigma}{min} (A)$ relies on the selection of the model
    set $\{f^{(i)} \}_i$ as well as the other two parameters $\eta$
    and $\sigma$. In practice, all the zero singular
      values of $A$ is perturbed by the noise with a factor of $\sigma$.
    One could select the minimal non-zero principal
    singular value to be $\subscr{\sigma}{min} (A)=\min \setdef{
      \sigma_i(A) }{\sigma_i(A) > \sigma, i \in \until{p} } $.
       }
\end{remark}

\noindent \textbf{Online Procedure:} To summarize, our online learning
methodology is given in Algorithm table~\ref{alg:Optal}.
Our approach
leverages the adaptation of a dynamic ambiguity set, together with
\textit{a-priori} knowledge of $\vect{d}$, and learns model parameter
$\vect{\alpha}$, and characterizes the unknown $f$ online via $\mathcal{P}$.
{\begin{algorithm}[hpb]
 \floatname{algorithm}{ $\probP$-Learning}{}
    \caption{${\textrm{Learn}}(\mathcal{I}, \vect{d} )$}
\label{alg:Optal}
\begin{algorithmic}[1]
  \Require $\{ f^{(i)}\}_i$, $\beta$, $T_0$, $\sigma$, $\theta$ and $t = 1$;
  \Ensure Online $\vect{\alpha}$, $\hat{\probP}$, $\hat{\epsilon}$
\Repeat
\State Update data set $\mathcal{I}:=\mathcal{I}_t$ and knowledge $\vect{d}:=\vect{d}_t$;
\State Compute $\vect{\alpha}:=\vect{\alpha}_t$ as in~\eqref{eq:alpha};
\State Select $\hat{\probP}_{t+1}$ as in~\eqref{eq:empdistn} and $\hat{\epsilon}:=\hat{\epsilon}_t$ as in~\eqref{eq:tractable_epsilon};
\State Leverage $(\hat{\probP}_{t+1},\hat{\epsilon} )$ as characterization of $f$;
\State $t \leftarrow t+1$;
\Until time $t$ stops.
\end{algorithmic}
\end{algorithm}}

\section{Simulations}\label{sec:simulation}
In this section, we illustrate the previous results on a simple
vehicle example. Consider a vehicle driving under various road
conditions, where its control signal is derived in advance, according
to a path-planner in an ideal environment.

Our goal is to learn the real-time environment and estimate the system
states via our adaptive $\probP$-Learning algorithm.
Our vehicle is modeled as a differential-drive robot subject to
uncertainty, see~\cite{SML:06}:
\begin{equation*}
  \begin{aligned}
  {x}_{1}^{+}=& {x}_{1} + h \cos({x}_{3}) {u}_{1} + h{w}_{1} , \\
  {x}_{2}^{+}=& {x}_{2} + h \sin({x}_{3}) {u}_{1} + h {w}_{2}, \\
  {x}_{3}^{+}=& {x}_{3} - h {u}_{2} + h {w}_{3},
\end{aligned} %
\end{equation*}
\begin{equation}
\begin{aligned}
  {u}_{1}=& \frac{r}{2} ( v_{l} + v_{r} + e_{1}   ), \\
  {u}_{2}=& \frac{r}{2R} ( v_{l} - v_{r} + e_{2}   ),
\end{aligned} \label{eq:unicycle}
\end{equation}
where $\vect{x}:=(x_1,x_2,x_3) \in \real^2 \times [-\pi,\pi) \cong
\real \times \unitcircle $ stands for vehicle position and orientation
on the 2-D plane. We denote by $\vect{x}^{+}$ the state at the next
time step and $\vect{w}:=(w_1,w_2,w_3)$ a zero-mean, mixture of
Gaussian and Uniform distributions, which are subGaussian
uncertainties
with $\sigma=0.5$.  We assume %
$\vect{x}_0=(0,0,0)$ and $h=10^{-3}$.  The velocity
$\vect{u}:=(u_1,u_2)$ is determined by a wheel radius $r=0.15$ m, the
distance between wheels $R=0.4$ m, the given wheel speed
$\vect{d}:=(v_{l}$, $v_{r})$ and an unknown parameter $
\vect{e}:=(e_{1} , e_{2})$, which depends on the wheel and road
conditions.  For simulation purposes, we assume that the vehicle may
move over three road zones, a slippery zone with
$\vect{e}^{(1)}=(4,0)$, a sandy zone with $\vect{e}^{(2)}=(-6,0)$, and
a smooth, regular zone with $\vect{e}^{(3)}=(0,0)$, as described in
Fig.~\ref{fig:vec_plan}.  The vehicle executes the following left and
right wheel speed plan (rad/s):
\begin{equation*} \small
  \begin{aligned}
    & v_{l} =10 -  0.5 \sin( 20 h \pi t ) , \\
    & v_{r} =10 +  0.5 \sin( 20 h \pi t ) .
  \end{aligned}
\end{equation*}
Now we employ our adaptive learning algorithm for the
characterization of the uncertain vehicle states and learning of the
unknown road-condition parameter $\vect{e}$ in real time. To do this,
we take $p=3$ predictors as in~\eqref{eq:unicycle} with
$\vect{w}\equiv 0$, and
\begin{equation*} \small
  \begin{aligned}
  &  i=1, \quad e_1=0, \; & e_2=0, \\
  &  i=2, \quad e_1=10, \; &  e_2=0, \\
  &  i=3, \quad e_1=0, \; & \quad e_2=10.
  \end{aligned}
\end{equation*}
Note that Assumption~\ref{assump:predictor} holds with
$\vect{\alpha}^{\star}:=(0.6,0.4,0)$ in the slippery zone,
$\vect{\alpha}^{\star}:=(1.6,-0.6,0)$ in the sandy zone and
$\vect{\alpha}^{\star}:=(1,0,0)$ in the smooth zone. We select
$T_0=300$, and, at each time $t$, we have access to model sets
$\{f^{(i)} \}_i$ as well as the real-time data set $\mathcal{I}_t$ and
$\vect{d}$. Note that the notions of inner product and norm are those defined %
on the vector space $T(\real^2 \times \mathbb{S}) \equiv \real^3$. Recall that $h= 10^{-3}$, so a $T_{0} = 300$ corresponds to a time window of order 0.3sec.
We select online diagonal regularization
matrices $P$ with diagonal ${\small (1/(\sqrt{3}\max_{i =1,2,3}
  |f^{(i)}(j)|)}$ for $j=1,2$ and $1$ for $j=3$, resulting in ${\small
  \eta= \max_{i,k\in \mathcal{T}} \Norm{ P_k f^{(i)}_{k}}} $.
\begin{figure}[tbp]%
\centering
\includegraphics[width=0.3\textwidth]{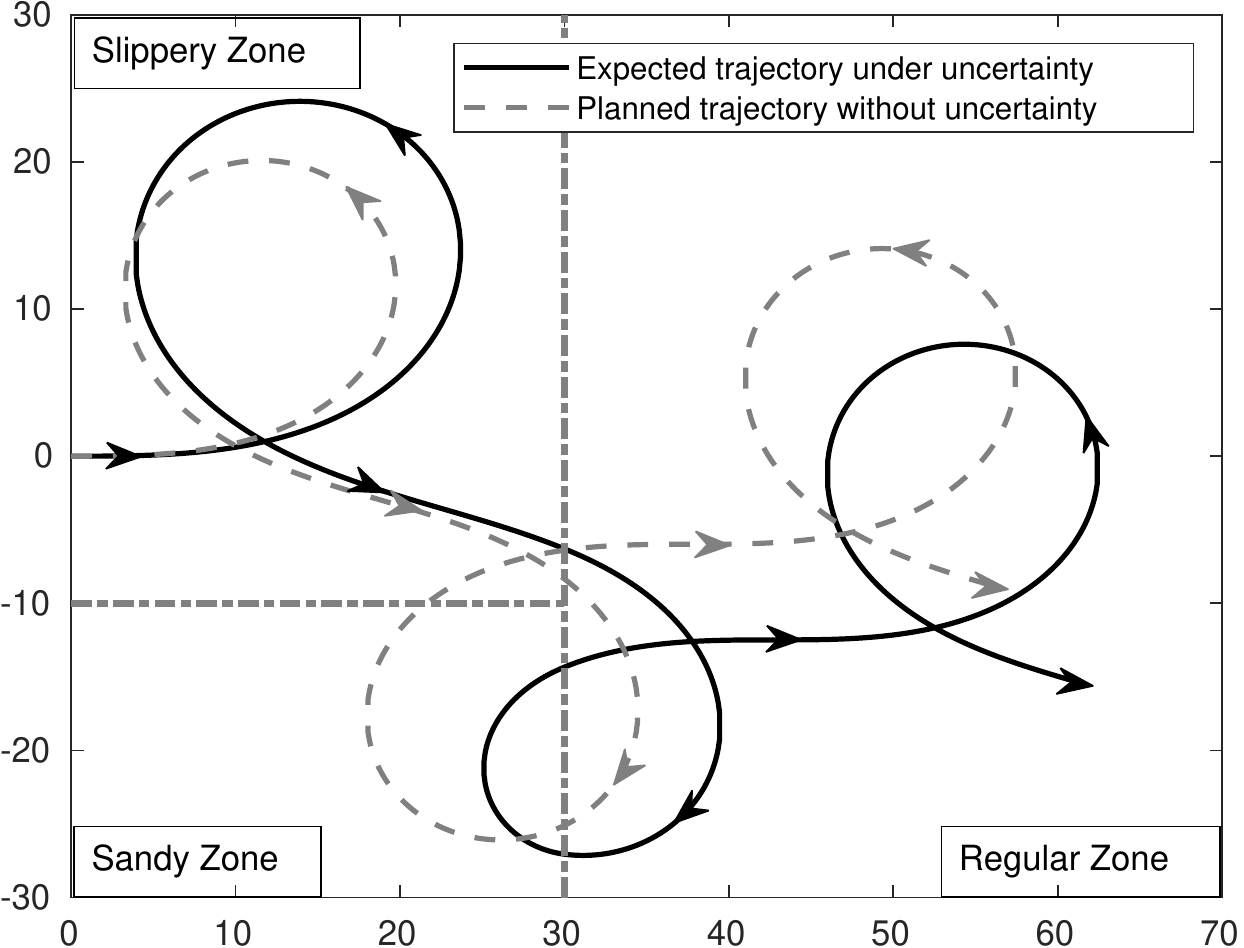}%
\caption{\small Path plan and actual trajectory in various $\real^2$ road zones.}%
\label{fig:vec_plan}%
\end{figure}
\begin{figure}[tbp]%
\centering
\includegraphics[width=0.225\textwidth]{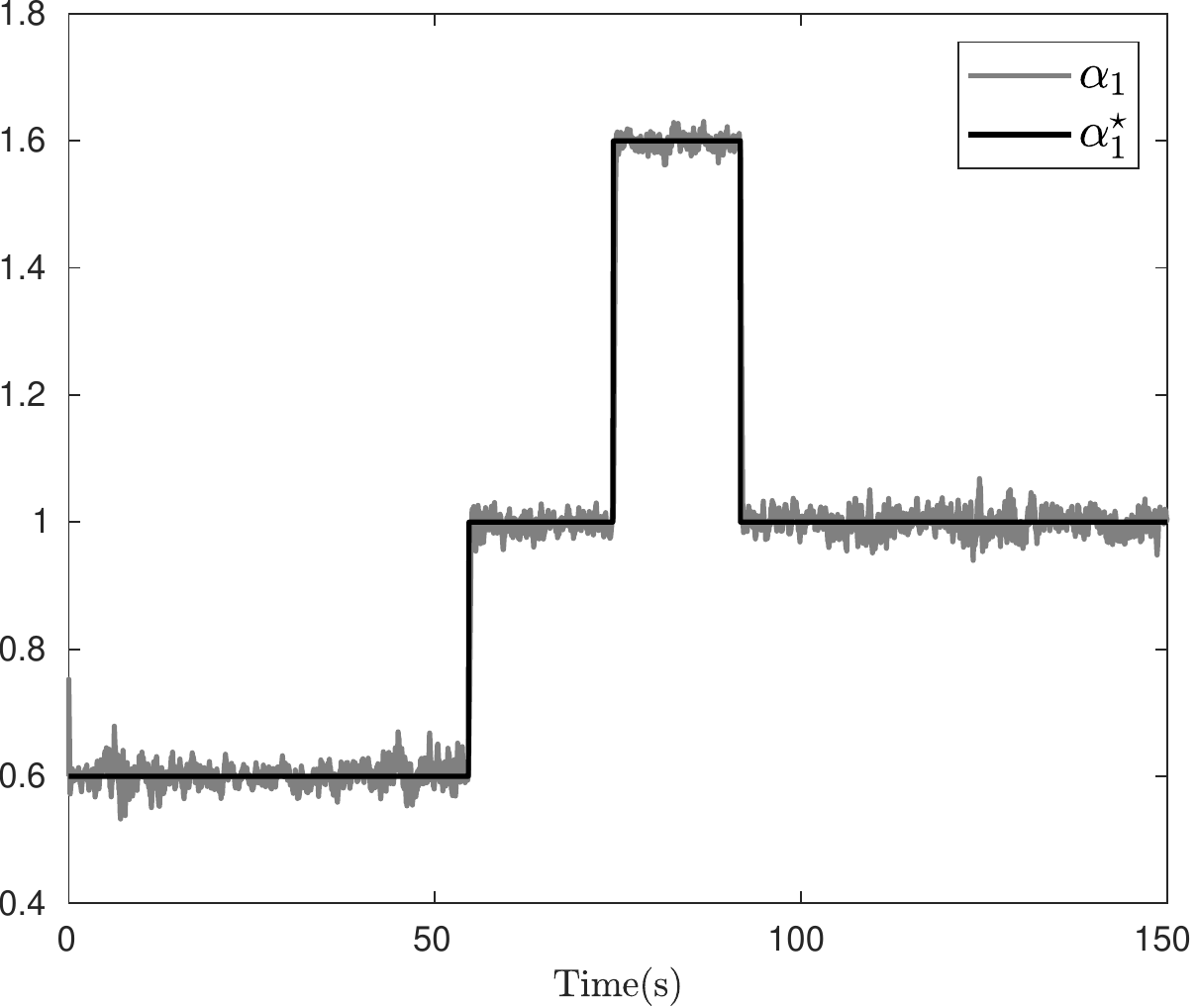}%
\includegraphics[width=0.24\textwidth]{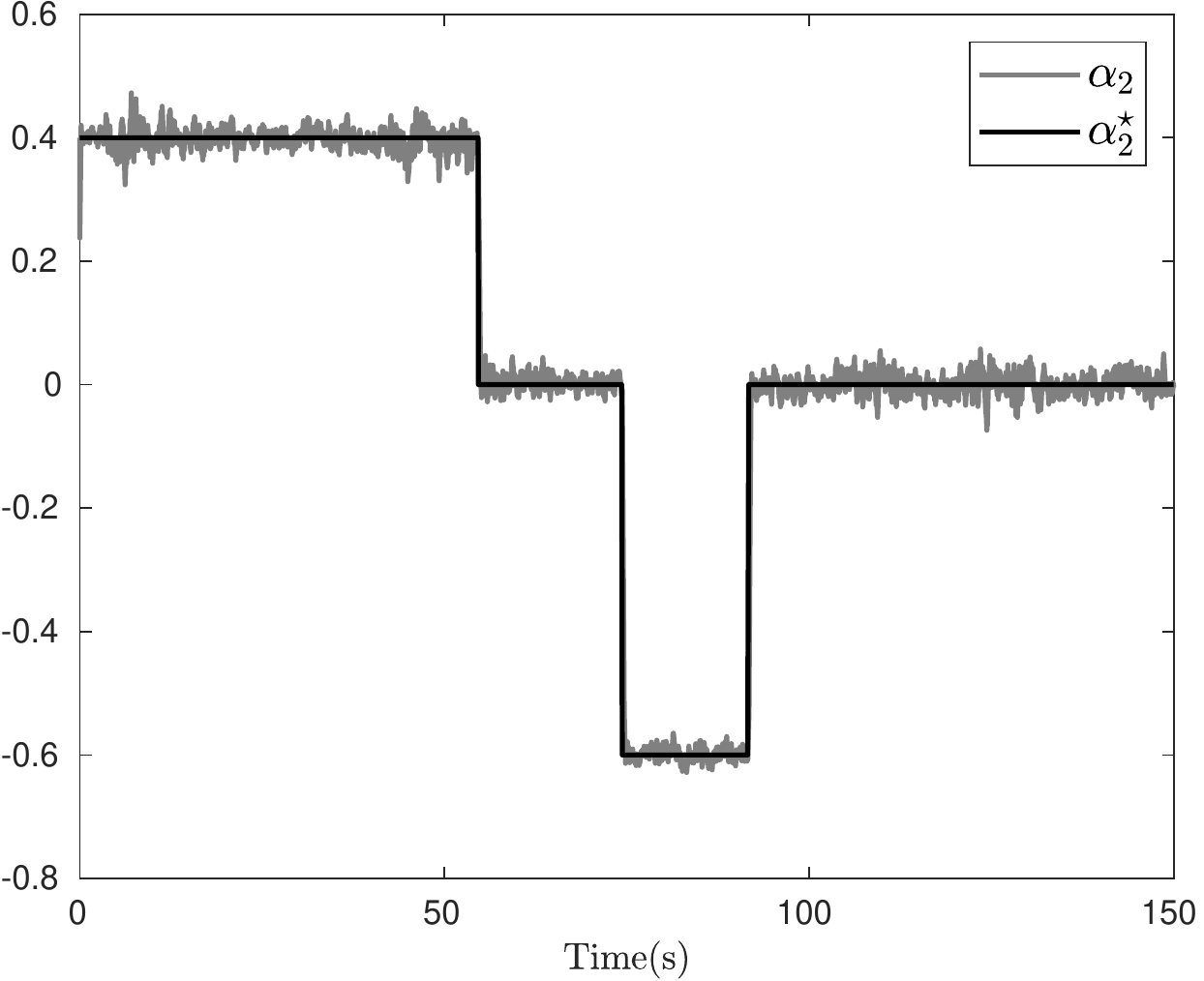}%
\caption{\small Real-time learning parameter $\alpha_1$ and $\alpha_2$.}%
\label{fig:vec_alpha}%
\end{figure}
\begin{figure}[tp]%
\centering
\includegraphics[width=0.235\textwidth]{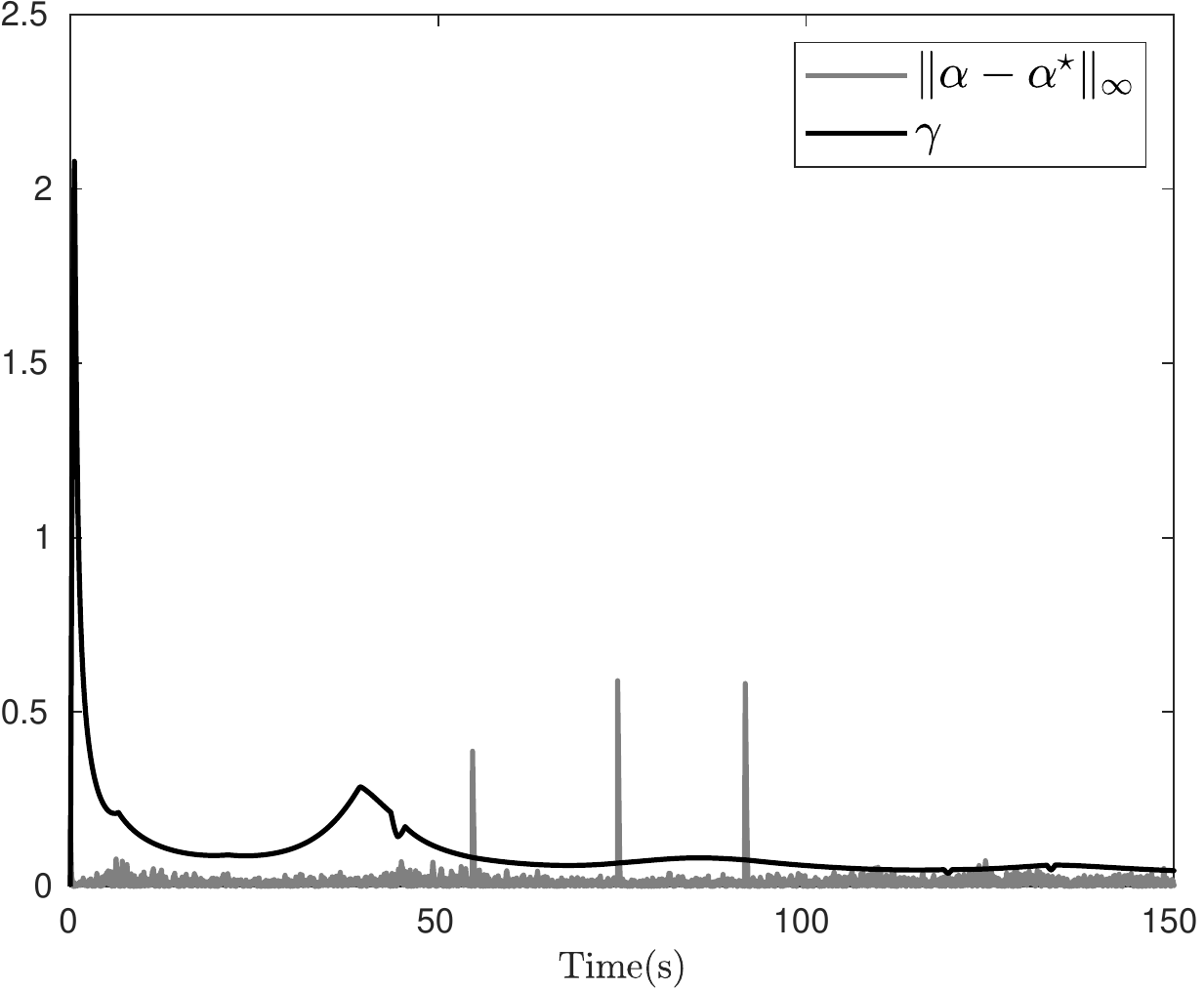}%
\includegraphics[width=0.237\textwidth]{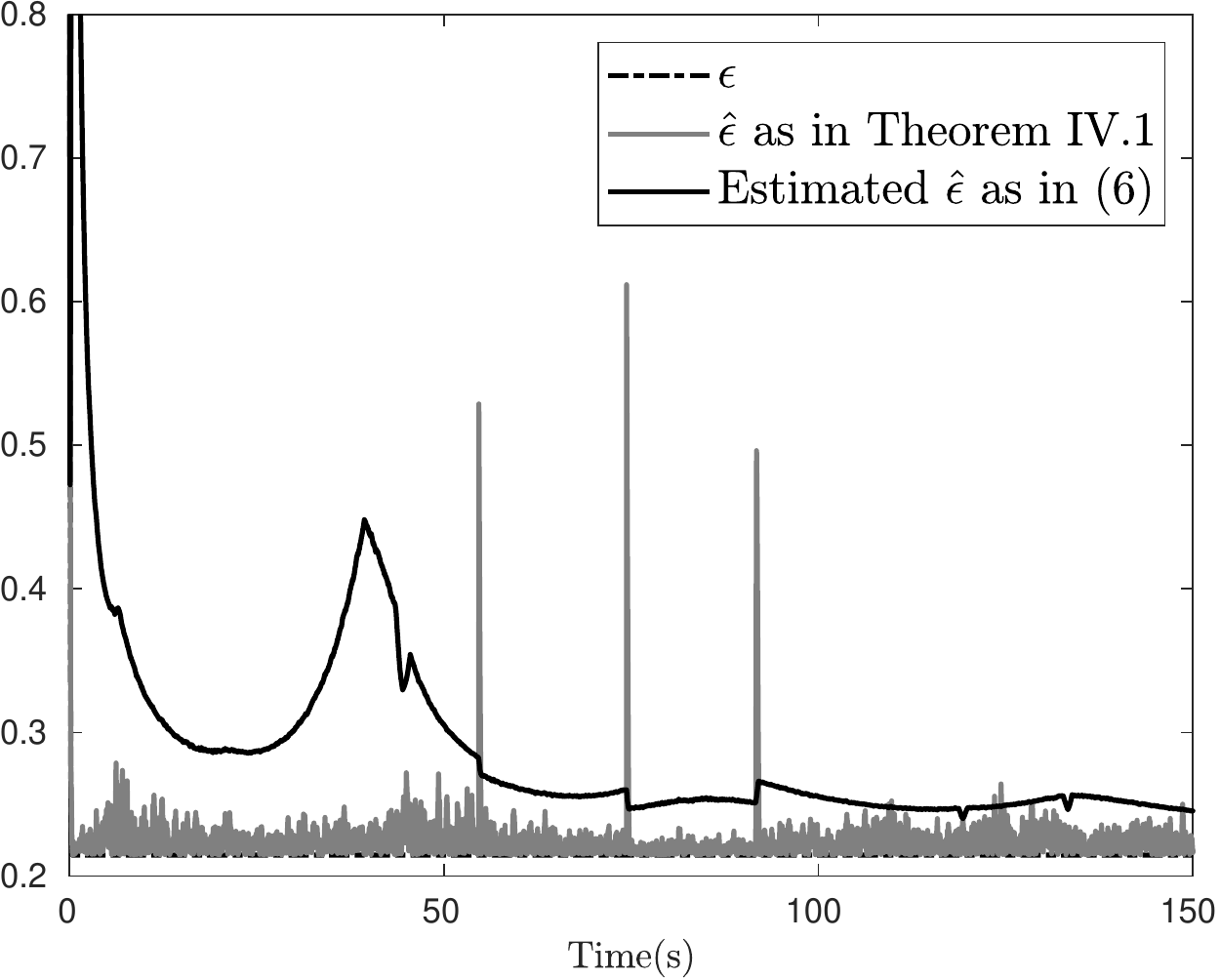}%
\caption{\small Quality of $\vect{\alpha}$ and the estimated radius  $\hat{\epsilon}$.}%
\label{fig:vec_epsilon}%
\end{figure}
\begin{figure}[tp]%
\centering
\includegraphics[width=0.235\textwidth]{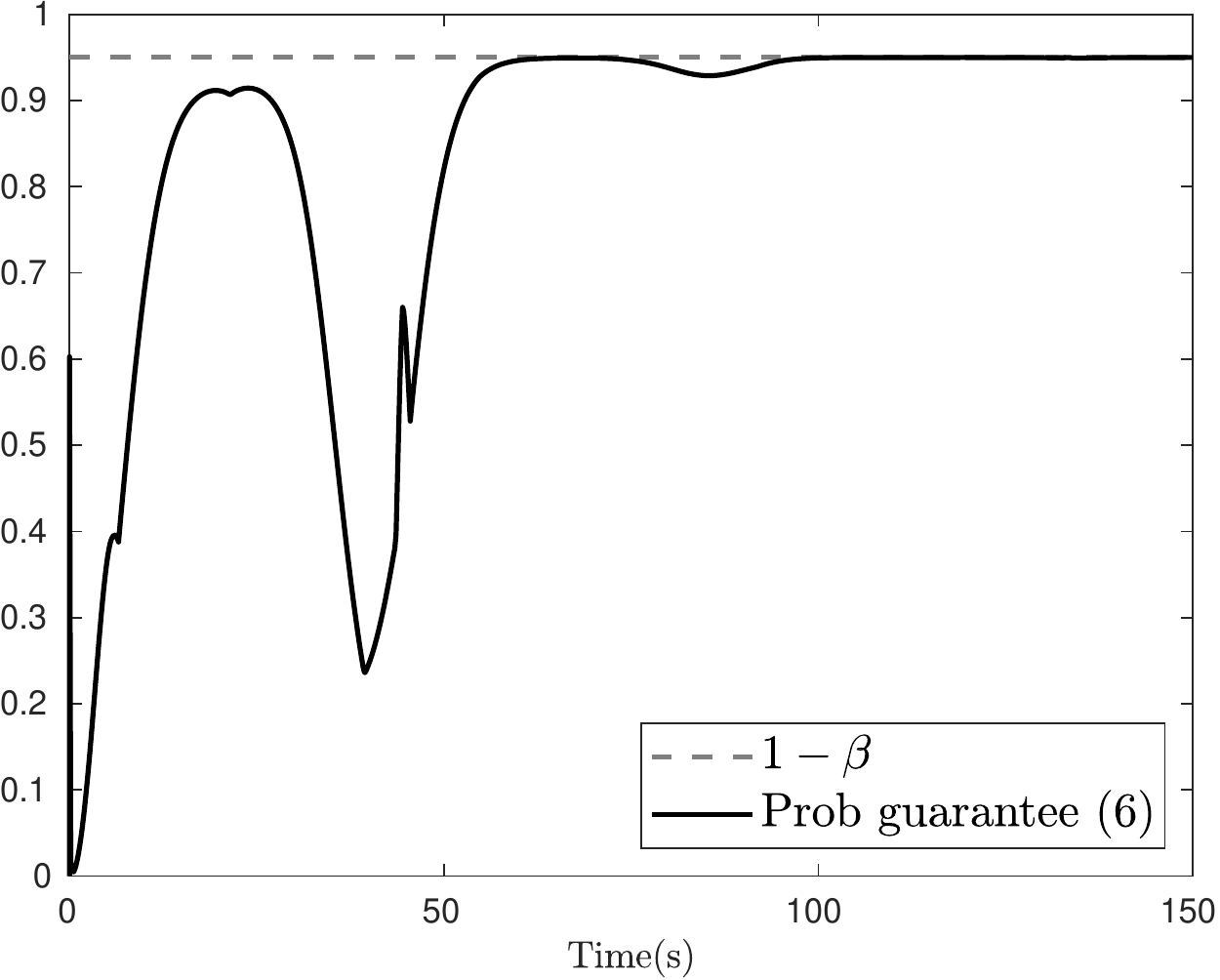}%
\includegraphics[width=0.235\textwidth]{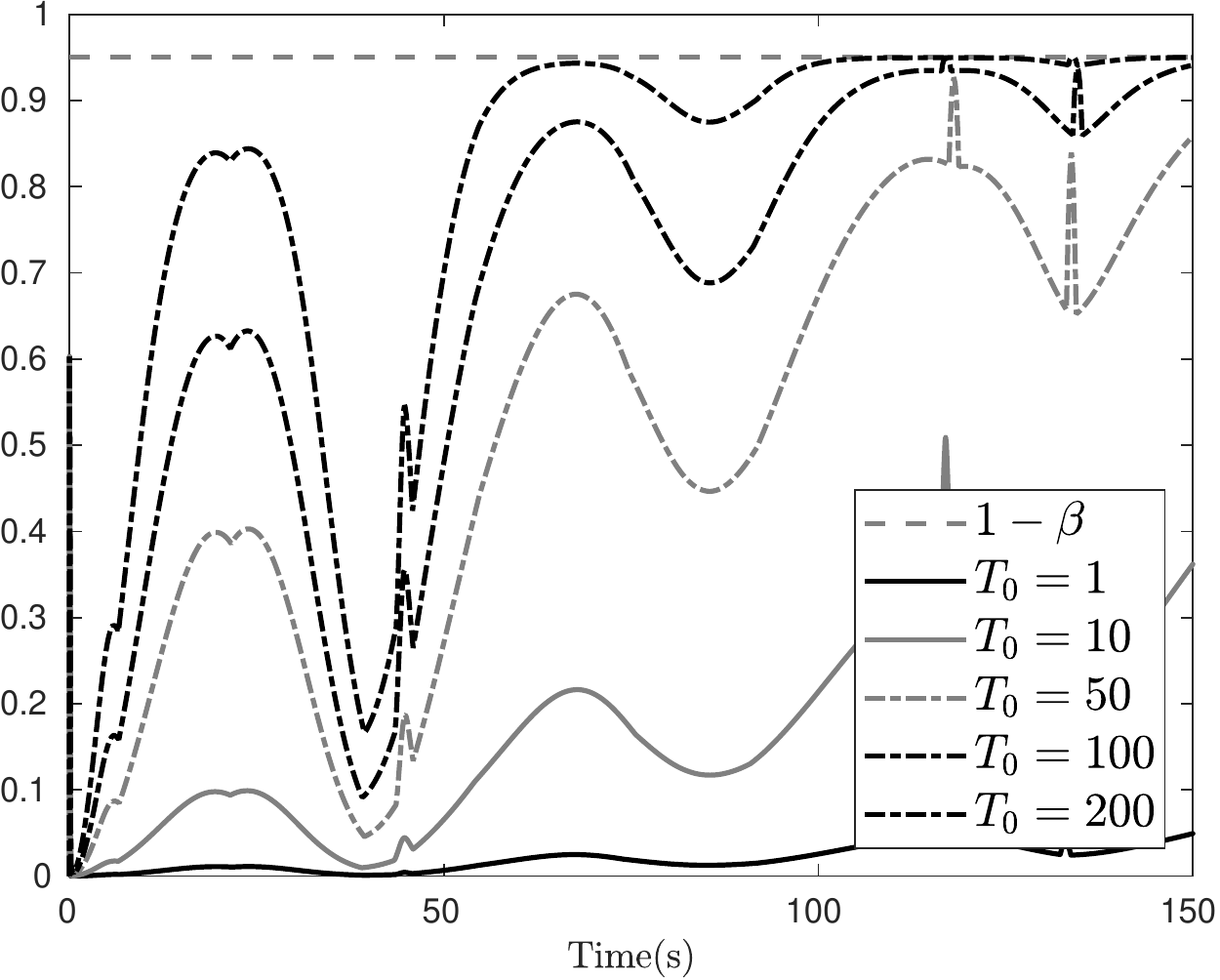}%
\caption{\small Online guarantee~\eqref{eq:modgua} and samples of~\eqref{eq:modgua} with various $T_0$.}%
\label{fig:modgua}%
\end{figure}

Fig.~\ref{fig:vec_alpha} demonstrates the real-time parameter learning
of $\alpha_1$ and $\alpha_2$. It can be seen that these unknown
parameters are effectively learned and tracked over
time.
Fig.~\ref{fig:vec_epsilon} shows the quality of the learned parameter
$\vect{\alpha}$ and its effect on the determination of the radius of
the adaptive ambiguity set.  We note that, for a particular noise
realization sequence, the estimated value $\gamma=n\sigma e \eta
\sqrt{np} \sigma_{\min}^{-1} (A) +\theta$, with $\theta=0.01$, upper
bounds $\Norm{\vect{\alpha}-\vect{\alpha}^{\star}}_{\infty}$ in high
probability. The large spikes in the figure are due to the change of
zone, resulting in a large error. This is expected, as the true
$\vect{\alpha}^\star$ changed discontinuously.
Meanwhile, the estimated radius $\hat{\epsilon}$ of the adaptive
ambiguity set, calculated as in~\eqref{eq:tractable_epsilon}, is a
conservative estimate of the unknown \textit{a-priori}
$\hat{\epsilon}$ as in Theorem~\ref{thm:ambiguityP}. The true
$\hat{\epsilon}$ captures exactly the ambiguity set over the time
sequence $\mathcal{T}$, for a $\beta= 0.05$. Over time, we empirically
see the difference between the approximated $\hat{\epsilon}$ via
$\gamma$ and the true one become close. In practice, the radius
$\hat{\epsilon}$ can be selected in a data-driven fashion, e.g., as in
Remark~\ref{remark:select_eps}, to serve as a way for less
conservative estimation of the radius in probability. We show in
Fig.~\ref{fig:modgua} the online guarantee~\eqref{eq:modgua} of this
particular case study, and various samples of~\eqref{eq:modgua},
obtained by taking different time horizon $T_0$.
\section{Conclusions}\label{sec:Conclude}
In this paper, we proposed an approach for online learning of unknown
and uncertain dynamical environments in a parameterized class.  The
proposed method allows us to learn the environment, while providing an
online characterization of the approximation via online-quantifiable
probabilistic guarantees. The approach opens a way for the robust
integration of the online learning with control design.  A robotic
example was used to demonstrate the efficacy of the method.

\appendix \label{appx:whole}

We adapt these two subGaussian properties for proofs.
\begin{lemma}[$\infty-$norm of subGaussian vectors have subGaussian
  tails{~\cite{RV:18}}] \label{lemma:inftytail} If
  Assumption~\ref{assump:subG} holds, then for each $k=0,1,2,\ldots$
  and any $\eta \geq 0$, we have
\begin{equation*}
  \Prob \left( \Norm{\vect{w}_k}_{\infty}
    \geq \eta  \right) \leq 2n \exp \left( - \frac{\eta^2}{2 \sigma^2} \right).
\end{equation*}
\end{lemma}
\begin{proof}{Lemma~\ref{lemma:inftytail}}
  Let $w_{k,i}$ denote the $\supscr{i}{th}$ component of
  $\vect{w}_k$ where $i \in \untilinterval{1}{n}$. We apply the
  definition of $\infty-$norm as the following
\begin{equation*}
  \begin{aligned}
  \Prob \left( \Norm{\vect{w}_k}_{\infty} \geq \eta  \right) &= \Prob \left( \max\limits_{i \in \untilinterval{1}{n}} | {w_{k,i}}| \geq \eta  \right), \\
  &= 1- \Prob \left( | {w_{k,i}}| \leq \eta, \; \forall i \in \untilinterval{1}{n}  \right).
  \end{aligned}
\end{equation*}
By the independence of $w_{k,i}$ as in Assumption~\ref{assump:subG}, we have
\begin{equation*}
  \begin{aligned}
  &\Prob \left( | {w_{k,i}}| \leq \eta, \; \forall i \in \untilinterval{1}{n}  \right) \\
  & \hspace{12ex} =\Prob \left( | {w_{k,1}}| \leq \eta  \right) \cdots \Prob \left( | {w_{k,n}}| \leq \eta  \right).
  \end{aligned}
\end{equation*}
Then for each $i \in \untilinterval{1}{n}$, we have\footnote{An equivalent representation of Assumption~\ref{assump:subG}: For any $\eta \geq 0$, $\Prob \left( | \trans{a} w_t | \geq \eta  \right) \leq 2 \exp \left( - \frac{\eta^2}{2 \Norm{a}^2 \sigma^2 } \right)$.}
\begin{equation*}
  \begin{aligned}
 \Prob \left( | {w_{k,i}}| \leq \eta  \right) & =1- \Prob \left( | {w_{k,i}}| \geq \eta  \right), \\
 & \geq 1- 2 \exp \left( - \frac{\eta^2}{2 \sigma^2}  \right),
  \end{aligned}
\end{equation*}
which results in
\begin{equation*}
  \begin{aligned}
  &\Prob \left( | {w_{k,i}}| \leq \eta, \; \forall i \in \untilinterval{1}{n}  \right)   \geq \left[ 1- 2 \exp \left( - \frac{\eta^2}{2 \sigma^2}  \right)  \right]^n.
  \end{aligned}
\end{equation*}
Finally, we have\footnote{Bernoulli's inequality: $(1+x)^n \geq 1 + nx$ for $\forall n \in \integerpositive$, $\forall x \geq -2$.}
\begin{equation*}
  \begin{aligned}
  \Prob \left( \Norm{\vect{w}_k}_{\infty} \geq \eta  \right) & \leq 1- \left[ 1- 2 \exp \left( - \frac{\eta^2}{2 \sigma^2}  \right)  \right]^n, \\
& \leq 2n \exp \left( - \frac{\eta^2}{2 \sigma^2} \right).
  \end{aligned}
\end{equation*}
\end{proof}

\begin{lemma}[Bounded moments of normed-subGaussian
  vectors{~\cite{RV:18}}] \label{lemma:inftymoment} If
  Assumption~\ref{assump:subG} holds,
\begin{equation*}
\mathbb{E}\left[ \Norm{\vect{w}_k}_{\infty}^{l} \right] \leq n
\sigma^{l} l^{\frac{l}{2}+1}, \;\; \forall \, l \in
  \integernonnegative.
\end{equation*}
\end{lemma}
\begin{proof}{Lemma~\ref{lemma:inftymoment}}
The moments can be equivalently computed by
\begin{equation*}
\mathbb{E}\left[ \Norm{\vect{w}_k}_{\infty}^{l} \right] = \int\limits_{0}^{\infty} \Prob \left(   \Norm{\vect{w}_k}_{\infty} \geq \eta \right) l \eta^{l-1} d \eta.
\end{equation*}
Applying the result of Lemma~\ref{lemma:inftytail}, we have
\begin{equation*}
\mathbb{E}\left[ \Norm{\vect{w}_k}_{\infty}^{l} \right] \leq 2nl
\int\limits_{0}^{\infty} \exp \left( - \frac{\eta^2}{2 \sigma^2} \right)
\eta^{l-1} d \eta.
\end{equation*}
By the variable substitute $\bar{\eta}:= \frac{\eta^2}{2 \sigma^2}$, the above bound becomes
\begin{equation*}
\mathbb{E}\left[ \Norm{\vect{w}_k}_{\infty}^{l} \right] \leq nl (2 \sigma^2)^{\frac{l}{2}}
\int\limits_{0}^{\infty} \exp \left( - \bar{\eta} \right)
\bar{\eta}^{\frac{l}{2}-1} d \bar{\eta}.
\end{equation*}
By the definition of $\Gamma$ function and its property\footnote{The property of the $\Gamma$ function: $\Gamma(\frac{l}{2}):=\int\limits_{0}^{\infty} \exp \left( - \bar{\eta} \right) \bar{\eta}^{\frac{l}{2}-1} d \bar{\eta}  \leq (\frac{l}{2})^{\frac{l}{2}}$.}, we have
\begin{equation*}
\mathbb{E}\left[ \Norm{\vect{w}_k}_{\infty}^{l} \right] \leq n \sigma^{l} l^{\frac{l}{2}+1}.
\end{equation*}
\end{proof}

\begin{myproof}{Lemma~\ref{lemma:perfectenv}}
  We prove this in two steps. First, we exploit the properties of
  $\vect{w}_{t}$. Then, we show the probabilistic guarantees of the
  dynamic ambiguity set.

  \noindent \textbf{Step 1 (subGaussian Wasserstein distances):}
  Given Assumption~\ref{assump:subG}, on the subGaussian $\probW_t$,
  the following holds:

  \noindent \textbf{(a)} Following~\cite[Lemma 1]{JNW-PR:19}, the
  distribution $\probW_t$ satisfies
\begin{equation}
  \begin{aligned}
    d_W( \hat{\probW}_t, \probW_t ) \leq \sqrt{ 2n \sigma^2
      \mathcal{D}(\hat{\probW}_t | \probW_t) }, \quad \forall \;
    \hat{\probW}_t \in \mathcal{M}(\real^n),
  \end{aligned} \label{eq:Wproperty}
\end{equation}
where $d_W$ and $\mathcal{D}$ denote the $1$-Wasserstein metric and
the KL divergence of two distributions $\hat{\probW}_t$
and $\probW_t$, respectively; and the set $\mathcal{M}(\real^n)$ is the
space of all probability distributions supported on $\real^n$
with a finite first moment.

\noindent \textbf{(b)} Let us denote $\hat{\probW}_{t}:=
\frac{1}{T} \sum_{k\in \mathcal{T}}
\delta_{\{ \hat{\vect{x}}_{k+1} -f({k},\hat{\vect{x}}_{k},\hat{\vect{d}}_{k}) \}}$.
Note that, by the assumption that $\vect{w}_{k}$ is i.d. for $k \in \mathcal{T}$, $\hat{\probW}_{t}$ is the empirical distribution of $\vect{w}_{t}$.
Then, following~\cite[Theorem 6]{JNW-PR:19}, we claim that the
equation~\eqref{eq:Wproperty} holds if and only if the random variable
$d_W(\hat{\probW}_{t}, \probW_t)$ is
$\sqrt{n}\sigma/\sqrt{T}$-subGaussian for all $t$. Equivalently, for
all $t$ and any $\lambda \in \real$, we have
\begin{equation}
{\small
  \begin{aligned} \label{eq:subgW}
    \mathbb{E} \left[ \exp \left( \lambda \cdot ( d_W(\hat{\probW}_{t}, \probW_t) - \mathbb{E} [ d_W(\hat{\probW}_{t}, \probW_t) ]    ) \right)  \right]
     \leq \exp ( \frac{ n \lambda^2 \sigma^2}{2T}).
\end{aligned}
}
\end{equation}
\noindent \textbf{(c)} At each $t$, let us consider $C_t:= \mathbb{E}
[ d_W(\hat{\probW}_{t}, \probW_t) ]$. Following~\cite[Theorem
1]{NF-AG:15} and~\cite[Theorem 3.1]{JL:18},
  we claim that, for $n \ne 2$, there
exists a constant $c$, depending on Assumption~\ref{assump:subG}, such
that
\begin{equation*}
\begin{aligned}
C_t \leq c \cdot T^{- {1}/{\max \{ n,2\} } }.
\end{aligned}
\end{equation*}
In particular, we have $c=3^{3.5} \times 2^{10} \times \sigma^3$ when
$n=1$. When $n>3$, the parameter $c$ is calculated by\footnote{ These
  parameters are obtained based on $1-$Wasserstein metric and the
  third moment of $\probW$; the bound is slightly different when using
  the moment information with different order. For the case $n=2$, the
  bound introduces logarithm term of $T$,e.g., $ C_t \leq c(T /
  \log(T))^{- {1}/{\max \{ n,2\} }}$, we refer reader
  to~\cite{NF-AG:15} and~\cite{JL:18} for details. }
\begin{equation*}
  c:=  (1+\sqrt{2})(1+\sqrt{3})  \times  3^{3.5 - {1}/{n}} \times 2^{7} \times \sigma^3 \times n^{1.5}.
\end{equation*}
\noindent \textbf{Step 2 (Probabilistically-guaranteed dynamic
  ambiguity sets):} Knowing that the distributions $\probP_{t+1}$ and
$\probW_t$ obey the environment dynamics~\eqref{eq:env}, in other words, $\probP_{t+1} \equiv f(t,\vect{x}_{t},\vect{d}_{t}) + \probW_t$, holds for any deterministic $f$. Similarly, we have $\probQ_{t+1} \equiv f(t,\vect{x}_{t},\vect{d}_{t}) + \hat{\probW}_{t}$. Therefore, by the definition of the Wasserstein metric, we claim that
\begin{equation*}
  d_W(\probQ_{t+1} %
  , \probP_{t+1}) \equiv d_W(\hat{\probW}_{t}, \probW_t), \quad \forall \vect{d} %
  , \; \forall t,
\end{equation*}
where the empirical distribution
$\probQ_{t+1}$ %
is described as in the statement of the lemma. Then by the Markov
inequality, %
for any $\gamma \geq 0$ and $\lambda \geq 0$, we have
\begin{equation*}
{\small
  \begin{aligned}
  \Prob & \left( d_W(\probQ_{t+1} %
  , \probP_{t+1}) \geq \gamma  \right) %
 =   \Prob \left( d_W(\hat{\probW}_{t}, \probW_t) \geq \gamma  \right)  \\
  & \hspace{6ex} =   \Prob \left( \exp ( \lambda \cdot d_W(\hat{\probW}_{t}, \probW_t) ) \geq \exp (\gamma \lambda)  \right) , \\
 & \hspace{6ex} \leq
 \exp \left( -\gamma \lambda \right) \mathbb{E}\left[
   \exp \left( \lambda \cdot d_W(\hat{\probW}_{t}, \probW_t) \right) \right].
  \end{aligned}}
\end{equation*}
Then by the property~\eqref{eq:subgW}, we have
\begin{equation*}
{\small
  \begin{aligned}
    & \Prob \left( d_W(\probQ_{t+1} %
      , \probP_{t+1}) \geq \gamma  \right) %
    \leq \exp \left( - (\gamma -C_t) \lambda + \frac{
        n \lambda^2 \sigma^2}{2T} \right),
  \end{aligned}}
\end{equation*}
where $C_t:= \mathbb{E} [ d_W(\hat{\probW}_{t}, \probW_t) ]$. The
optimal $\lambda \in \realnonnegative$ that results in the tightest bound
is taken to be
\begin{equation*}
\begin{aligned}
\lambda:= \begin{cases}
\frac{(\gamma -C_t) T}{n \sigma^2} ,  &\textrm{ if } \gamma > C_t , \\
    0 , &\textrm{ if } \gamma \leq C_t,
  \end{cases}
\end{aligned}
\end{equation*}
resulting in, when $\gamma > C_t$,
\begin{equation*}
  \begin{aligned}
  \Prob & \left( d_W(\probQ_{t+1} %
  , \probP_{t+1}) \geq \gamma  \right) \leq
 \exp \left( - \frac{(\gamma -C_t)^2 T}{2 n \sigma^2} \right) .
  \end{aligned}
\end{equation*}
Finally, let $\gamma > C_t$  be
\begin{equation*}
\begin{aligned}
  \gamma=\epsilon:=\epsilon(T,\beta)= %
   \sqrt{ \frac{2n \sigma^2}{T} \ln(\frac{1}{\beta})}  + c \cdot T^{- {1}/{\max \{ n,2\} } },
\end{aligned}
\end{equation*}
where $c$ is determined as in  step 1. This results in
\begin{equation*}
  \begin{aligned}
  \Prob & \left( d_W(\probQ_{t+1} %
  , \probP_{t+1}) \geq \epsilon %
    \right) \leq
 \beta,
  \end{aligned}
\end{equation*}
and further, we have
\begin{equation*}
  \begin{aligned}
  \Prob & \left( d_W(\probQ_{t+1} %
  , \probP_{t+1}) \leq \epsilon %
    \right) \geq
1- \beta.
  \end{aligned}
\end{equation*}
Equivalently, we have~\eqref{eq:guarantee} by selecting $
\mathcal{P}_{t+1} %
:= \mathbb{B}_{\epsilon %
}(\probQ_{t+1} %
).$
If we take $T_0=\infty$, we have $T=t$ with $t \rightarrow
\infty$. Then, it obviously follows that
$\mathcal{P}_{t+1}$ %
shrinks to $\probP_{t+1}$ as $t \rightarrow \infty$. $\hfill$
\qed \end{myproof}

\begin{myproof}{Theorem~\ref{thm:ambiguityP}}
  From the triangular inequality,
  \begin{equation*}
    \subscr{d}{W}(\probP_{t+1},\hat{\probP}_{t+1}) \leq \subscr{d}{W}(\probP_{t+1},\probQ_{t+1}) + \subscr{d}{W}(\probQ_{t+1},\hat{\probP}_{t+1}),
  \end{equation*}
  where by Lemma~\ref{lemma:perfectenv}, we have %
\begin{equation*}
  {\Prob}\left( \subscr{d}{W}(\probP_{t+1},\probQ_{t+1}) \leq \epsilon %
   \right)\geq 1- \beta.
\end{equation*}
To evaluate the second term above, we apply the definition of
Wasserstein metric, given as
\begin{equation*}
  \subscr{d}{W}(\probQ_{t+1},\hat{\probP}_{t+1})=\sup_{h \in
    \mathcal{L}}  \int_{\mathcal{Z}} {h(\xi)\probQ_{t+1}(d{\xi})}
  -\int_{\mathcal{Z}} {h(\xi)\hat{\probP}_{t+1}(d{\xi})},
\end{equation*}
where the set $\mathcal{Z}$ is the support of the random variable
$\vect{x}_{t+1}$ and the set $\mathcal{L}$ is the space of all
Lipschitz functions defined on $\mathcal{Z}$ with Lipschitz constant
1.  Then, we equivalently
write %
\begin{equation*}
{\small
  \begin{aligned}
  &  \subscr{d}{W}(\probQ_{t+1},\hat{\probP}_{t+1}) = \sup_{h \in
      \mathcal{L}}  \left\{ \frac{1}{T} \sum\limits_{k\in \mathcal{T}} \left( h({\xi}_{k})- h(\sum\limits_{i=1}^{p} \alpha_i \xi_k^{(i)}) \right) \right\} , \\
      & \leq \sup_{h \in
        \mathcal{L}} \left\{ \frac{1}{T} \sum\limits_{k \in \mathcal{T}} | h({\xi}_{k})- h(\sum\limits_{i=1}^{p} \alpha_i \xi_k^{(i)}) | \right\} \\
      & \leq  \frac{1}{T} \sum\limits_{k \in \mathcal{T}}
      \Norm{ {\xi}_{k}  -\sum\limits_{i=1}^{p} \alpha_i \xi_k^{(i)} } \\
      & =  \frac{1}{T} \sum\limits_{k \in \mathcal{T}}
      \| \sum\limits_{i=1}^{p}  (\alpha_i -\alpha^{\star}_i) ( f^{(i)}({k},\hat{\vect{x}}_{k},\hat{\vect{d}}_{k}) - f^{(i)}({t},\hat{\vect{x}}_{t}, \vect{d})   )  \| , \\
& \leq  %
\sum\limits_{i=1}^{p} \left| \alpha^{\star}_i -\alpha_i \right|
\left[ \frac{1}{T} \sum\limits_{k \in \mathcal{T}} \Norm{ f^{(i)}({k},\hat{\vect{x}}_{k},\hat{\vect{d}}_{k}) - f^{(i)}({t},\hat{\vect{x}}_{t}, \vect{d}) } \right], \\
& \leq  %
\Norm{\vect{\alpha}^{\star} -\vect{\alpha}}_{\infty} H(t,T,{\vect{d}}), \\
  \end{aligned}
}
\end{equation*}
where
the first line comes from the Wasserstein distance between discrete
distributions $\probQ_{t+1}$ and $\hat{\probP}_{t+1}$; the second line
is followed by adding absolute operation and applying triangular
inequality; the third line comes from the definition of the set
$\mathcal{L}$; the fourth line is from the definition of the
environment predictions and Assumption~\ref{assump:predictor}, on the
representation of unknown environment $f$; the fifth one applies
triangular inequality, and the last line uses the Hölder's
inequality. Note that the derived bound of
$\subscr{d}{W}(\probQ_{t+1},\hat{\probP}_{t+1})$ holds true with
probability one.  Then, %
by summing the probability bounds of the two terms, we obtain
\begin{equation*}
  {\Prob}\left( \subscr{d}{W}(\probP_{t+1},\hat{\probP}_{t+1}) \leq \hat{\epsilon} %
  \right)\geq 1- \beta,
\end{equation*}
which can be written as~\eqref{eq:guarantee} with $\mathcal{P}_{t+1} %
:= \mathbb{B}_{\hat{\epsilon} %
}(\hat{\probP}_{t+1} %
)$. \hfill \qed
\end{myproof}

\begin{myproof}{Theorem~\ref{thm:proj}}
  To see this, first, we will bound $\Norm{\vect{\alpha} -
    \vect{\alpha}^{\star}}_{\infty}$ by samples
  of %
  $\vect{w}_{k}$, ${k \in \mathcal{T}}$, then we apply concentration
  results for probabilistic bounds on $\Norm{\vect{\alpha} -
    \vect{\alpha}^{\star}}_{\infty}$.

  \noindent \textbf{Step 1 (Bound on $\Norm{\vect{\alpha} -
      \vect{\alpha}^{\star}}_{\infty}$):} At
  each %
  $k \in \mathcal{T}$, let us denote by $\hat{\vect{w}}_k$ a sample of
  $\vect{w}_k$ represented by
\begin{equation*}
  \hat{\vect{w}}_k:=\hat{\vect{x}}_{k+1} - \sum\limits_{j=1}^{p} \alpha^{\star}_{j}f^{(j)}_k %
  .
\end{equation*}
Then, we project the data $\hat{\vect{x}}_{k+1}$ on the direction of
each regularized predictor $i \in \until{p}$,
\begin{equation*}
\begin{aligned}
  &\langle \hat{\vect{x}}_{k+1} , P_k
  f^{(i)}_k %
  \rangle %
  = \langle \sum\limits_{j=1}^{p}
  \alpha^{\star}_{j}f^{(j)}_k %
  + \hat{\vect{w}}_k , P_k
  f^{(i)}_k %
  \rangle,
\end{aligned}
\end{equation*}
where, given a scalar $\eta>0$, the time-dependent regularization
matrix is selected so that $\Norm{ P_k f^{(i)}_{k}} \leq \eta$, for
all $i \in \until{p}$.
Averaging the above equalities over $k\in \mathcal{T}$, we have for
each component $i$:
\begin{equation*}
\begin{aligned}
  \vect{b}(i) = \sum\limits_{j=1}^{p} \alpha^{\star}_{j} A(i,j) +
  \frac{1}{T}\sum\limits_{k\in \mathcal{T}}\langle \hat{\vect{w}}_k ,
  P_k f^{(i)}_k %
\rangle,
\end{aligned}
\end{equation*}
where
\begin{equation*}
  \begin{aligned}
\vect{b}(i):= \frac{1}{T}\sum\limits_{k\in \mathcal{T}} \langle \hat{\vect{x}}_{k+1}  , P_k f^{(i)}_{k} \rangle, \;\;  i \in \until{p},
    \end{aligned}
\end{equation*}
\begin{equation*}
  \begin{aligned}
  A(i,j):= \frac{1}{T} \sum\limits_{k\in \mathcal{T}} \langle f^{(j)}_{k} ,  P_k  f^{(i)}_{k}  \rangle, \; \; i,\; j \in \until{p}. \\
    \end{aligned}
\end{equation*}
By selecting $\vect{\alpha}$ as in~\eqref{eq:alpha}, the relation
$\vect{b}(i) = \sum_{j=1}^{p} \alpha_{j} A(i,j)$ holds, for each
$i$.
By subtracting the above equation from the one related to
$\vect{\alpha}^{\star}$, we have
\begin{equation*}
\begin{aligned}
  \sum\limits_{j=1}^{p} (\alpha_{j}- \alpha^{\star} ) A({i,j}) =
  \frac{1}{T} \sum\limits_{k \in \mathcal{T}}\langle \hat{\vect{w}}_k
  , P_k f^{(i)}_k %
  \rangle.
\end{aligned}
\end{equation*}
By taking the Moore–Penrose inverse of $A$, we obtain
\begin{equation*}
\begin{aligned}
\vect{\alpha} - \vect{\alpha}^{\star} = \pinverse{A} \vect{c},
\end{aligned}
\end{equation*}
where the vector $\vect{c}$ is
\begin{equation*}
  \frac{1}{T} \sum\limits_{k\in \mathcal{T}} \trans{\left( \langle  \hat{\vect{w}}_k  , P_k f^{(1)}_k %
      \rangle, \ldots, \langle  \hat{\vect{w}}_k  , P_k f^{(p)}_k %
      \rangle \right)}.
\end{equation*}
Take the $\infty$-norm operation on both sides, we have
\begin{equation*}
\begin{aligned}
\Norm{\vect{\alpha} - \vect{\alpha}^{\star}}_{\infty} \leq \Norm{\pinverse{A}}_{\infty} \Norm{\vect{c}}_{\infty},
\end{aligned}
\end{equation*}
where we can write  $\|\vect{c}\|_\infty$ as the following
\begin{equation*}
  \begin{aligned}
    \Norm{\vect{c}}_{\infty}:&= \frac{1}{T} \max\limits_{i\in
      \until{p}} \left\{ \left| \sum\limits_{k\in \mathcal{T}}\langle
        \hat{\vect{w}}_k , P_k
        f^{(i)}_k %
        \rangle \right| \right\}, \\
    & \leq \frac{1}{T}\max\limits_{i\in \until{p}} \left\{
      \sum\limits_{k\in \mathcal{T}} |\langle \hat{\vect{w}}_k , P_k
      f^{(i)}_k %
      \rangle | \right\} , \\
    & \leq \frac{1}{T}\max\limits_{i\in \until{p}} \left\{
      \sum\limits_{k\in \mathcal{T}}
      \left(\Norm{\hat{\vect{w}}_k} \cdot \Norm{ P_k
          f^{(i)}_k %
        } \right) \right\} , \\
    & \leq \frac{1}{T}\sum\limits_{k\in \mathcal{T}}
    \left(\Norm{\hat{\vect{w}}_k} \cdot \max\limits_{i\in
        \until{p}} \left\{ \Norm{P_k
          f^{(i)}_k %
        }  \right\} \right), \\
    & \leq \frac{\eta \sqrt{n}}{T}\sum\limits_{k\in \mathcal{T}}
    \Norm{\hat{\vect{w}}_k}_{\infty},
  \end{aligned}
\end{equation*}
where we achieve the first inequality by moving the absolute operation
into the sum operation; the second inequality uses
H\"older's %
inequality; the third inequality is achieved by moving max operation
into sum
operation; %
the forth one is achieved by %
the norm equivalence and the fact
that $\Norm{ P_k f^{(i)}_{k}} \leq \eta$ for all $i \in \until{p}$.
Then, we achieve the following bound
\begin{equation}
\begin{aligned}
  \Norm{\vect{\alpha} - \vect{\alpha}^{\star}}_{\infty} \leq \eta \sqrt{n}
  \Norm{\pinverse{A}}_{\infty} \left[ \frac{1}{T} \sum\limits_{k \in
      \mathcal{T}} \left( %
      \Norm{\hat{\vect{w}}_k}_{\infty} \right) \right].
\end{aligned} \label{eq:alphabd}
\end{equation}
Note that, by the equivalence of the matrix norm, we have
\begin{equation*}
  \Norm{\pinverse{A}}_{\infty} \leq \sqrt{p} \Norm{\pinverse{A}}_{2} =\sqrt{p} \; \subscr{\sigma}{max} (\pinverse{A})  \leq \sqrt{p}\; \subscr{\sigma}{min}^{-1} (A),
\end{equation*}
where $\subscr{\sigma}{max} (\pinverse{A})$ and $\subscr{\sigma}{min}(A)$ denote the
maximal singular value of $\pinverse{A}$ and the minimal principal
non-zero singular value of $A$, respectively.

\noindent \textbf{Step 2 (Measure concentration on
  $\Norm{\vect{\alpha} - \vect{\alpha}^{\star}}_{\infty}$):} In this
step, we find the probabilistic bound of $\Norm{\vect{\alpha} -
  \vect{\alpha}^{\star}}_{\infty}$ by developing that of
$\Norm{\vect{w}_k}_{\infty}$. Equivalently, given any $\gamma >0$, we
compute the following term
\begin{equation}
{\small
  \Prob \left( \frac{1}{T} \sum\limits_{k \in \mathcal{T}} \left( %
      \Norm{\vect{w}_k}_{\infty}  \right) \geq \gamma  \right).}
\label{eq:boundw}
\end{equation}
There are two options to obtain the bound. \\
\noindent \textbf{(1) (A naive bound via Markov inequality):} By the
Markov inequality, we obtain a bound~\eqref{eq:boundw} as
\begin{equation*}
{\small
  \Prob \left( \frac{1}{T} \sum\limits_{k \in \mathcal{T}} \left( %
   \Norm{\vect{w}_k}_{\infty}  \right) \geq \gamma  \right) \leq
    \frac{1}{\gamma T} \sum\limits_{k \in \mathcal{T}}
    \mathbb{E}\left[ \Norm{\vect{w}_k}_{\infty} \right].
}
\end{equation*}
By Lemma~\ref{lemma:inftymoment}, we have
$\mathbb{E}\left[ \Norm{\vect{w}_k}_{\infty} \right] \leq n \sigma,$
resulting in
\begin{equation*}
  \begin{aligned}
  &  \Prob \left( \Norm{\vect{\alpha} - \vect{\alpha}^{\star}}_{\infty}  \leq \gamma \right)  \geq
    1- \frac{1 }{\gamma} n \sigma \eta \sqrt{np}\subscr{\sigma}{min}^{-1} (A),
  \end{aligned}
\end{equation*}
with non-trivial bound if we take $\gamma > n \sigma \eta \sqrt{np}\; \subscr{\sigma}{min}^{-1} (A)$. \\
\noindent \textbf{(2) (A bound with exponential decay over $T$):}
For any $\lambda \geq 0$, the probability~\eqref{eq:boundw} is equivalent to
\begin{equation*} {\small \Prob \left( \exp \left( \sum\limits_{k\in
          \mathcal{T}} \left( \frac{\lambda %
          }{T} \Norm{\vect{w}_k}_{\infty} \right) \right) \geq \exp
      \left(\gamma \lambda \right) \right).}
\end{equation*}
By the Markov inequality to the above probablity, we have
\begin{equation*}
{\small
  \begin{aligned}
  \Prob & \left( \frac{1}{T} \sum\limits_{k \in \mathcal{T}} \left( %
  \Norm{\vect{w}_k}_{\infty}  \right) \geq \gamma  \right)  \\
 & \hspace{8ex}  \leq \exp \left( -\gamma \lambda \right) \mathbb{E}\left[ \prod\limits_{k\in \mathcal{T}}   \exp \left( \frac{\lambda %
 }{T} \Norm{\vect{w}_k}_{\infty} \right) \right].
  \end{aligned}}
\end{equation*}
By Assumption~\ref{assump:subG} on independence of $\vect{w}_k$, we
have
\begin{equation*}
  \mathbb{E}\left[ \prod\limits_{k\in \mathcal{T}}   \exp \left( \frac{\lambda %
      }{T} \Norm{\vect{w}_k}_{\infty} \right) \right] = \prod\limits_{k\in \mathcal{T}} \mathbb{E}\left[ \exp \left( \frac{\lambda %
      }{T} \Norm{\vect{w}_k}_{\infty} \right) \right].
\end{equation*}
For each %
$k \in \mathcal{T}$, we write each $\exp$ operation in its power
series form as the following
\begin{equation*}
  \begin{aligned}
    \mathbb{E}\left[ \exp \left( \frac{\lambda %
        }{T} \Norm{\vect{w}_k}_{\infty} \right) \right] &=
    \mathbb{E}\left[ 1+ \sum\limits_{l=1}^{\infty} \frac{ \left(
          \frac{\lambda %
          }{T} \right)^{l}  \Norm{\vect{w}_k}_{\infty}^{l} }{l!}\right] , \\
    &= 1+ \sum\limits_{l=1}^{\infty} \frac{ \left( \frac{\lambda %
        }{T} \right)^{l} \mathbb{E}\left[
        \Norm{\vect{w}_k}_{\infty}^{l} \right] }{l!}.
  \end{aligned}
\end{equation*}
By Lemma~\ref{lemma:inftymoment}, we have
\begin{equation*}
  \mathbb{E}\left[ \Norm{\vect{w}_k}_{\infty}^{l} \right] \leq n \sigma^{l} l^{\frac{l}{2}+1}, \; \forall l=1,2,\ldots.
\end{equation*}
This gives\footnote{We use two facts: 1) $l! \geq (l/e)^{l}$ and 2) $l^{1- \frac{l}{2}} \leq 1$, for all $l \in \integernonnegative$, where the constant $e=2.71828...$.}
\begin{equation*}
  \begin{aligned}
    \mathbb{E}\left[ \exp \left( \frac{\lambda %
    }{T} \Norm{\vect{w}_k}_{\infty} \right) \right]
& \leq 1+  n \sum\limits_{l=1}^{\infty}
 \left( \frac{ \lambda %
  \sigma e  }{T}\right)^{l}.  %
  \end{aligned}
\end{equation*}
To tighten the previous upper bound, consider any $\lambda$ such that
$ \lambda \in [0, \frac{T}{ %
 \sigma e}).$
Then the following bound holds\footnote{We use the fact: $1+x \leq
  \exp(x)$ for $x \in \real$.}
\begin{equation*}
{\small
    \mathbb{E}\left[ \exp \left( \frac{\lambda %
        }{T} \Norm{\vect{w}_k}_{\infty} \right) \right] %
\leq 1+
    \frac{ \lambda %
      \sigma n e }{T- \lambda %
      \sigma e} %
\leq \exp \left( \frac{ \lambda %
        \sigma n e }{T- \lambda %
        \sigma e} \right).}
\end{equation*}
Finally, we achieve
\begin{equation*}
{\small
  \Prob %
\left( \frac{1}{T} \sum\limits_{k \in \mathcal{T}} \left( %
   \Norm{\vect{w}_k}_{\infty}  \right) \geq \gamma  \right)
\leq \exp \left( -\gamma \lambda + \sum\limits_{k\in \mathcal{T}} \frac{ \lambda %
 \sigma n e }{T- \lambda %
  \sigma e}  \right) .
}
\end{equation*}
Finding an optimal bound is hard, and therefore we find a sub-optimal
bound by selecting $\lambda$ to be
\begin{equation*}
\lambda=
\begin{cases}
  \frac{T}{2 %
  \sigma e} -\frac{nT}{2 \gamma}, &\textrm{ if } \gamma \geq %
   \sigma n e , \\
    0, &\textrm{ if } \gamma < %
     \sigma n e .
\end{cases}
\end{equation*}
Then, we have the following
\begin{equation*}
  \begin{aligned}
  & \Prob \left( \frac{1}{T} \sum\limits_{k \in \mathcal{T}} \left( %
   \Norm{\vect{w}_k}_{\infty}  \right) \geq \gamma  \right)    \\
  & \hspace{0.4cm}  \leq \begin{cases}
    \exp \left( - \frac{ ( %
     \sigma n e - \gamma)^2 T^2 }{2 \left[ (2T-1)\gamma %
      \sigma e + n ( %
       \sigma e )^2 \right]} \right) ,  &\textrm{ if } \gamma \geq %
        \sigma n e , \\
    1 , &\textrm{ if } \gamma < %
     \sigma n e .
\end{cases}
  \end{aligned}
\end{equation*}
In words, the probability bounds on the quality of $\vect{\alpha}$ is
\begin{equation*}
{\small
\Prob \left( \Norm{\vect{\alpha} - \vect{\alpha}^{\star}}_{\infty}
  \leq \gamma \right)
\geq 1- \exp \left( - \frac{ (n c - \gamma)^2 T^2
      }{2 \left[ (2T-1)c \gamma + n c^2 \right]} \right),
}
\end{equation*}
with any $\gamma \geq n c$,
where %
$c:= %
\sigma e \eta \sqrt{np}\; \subscr{\sigma}{min}^{-1} (A)$.
 $\hfill$ \qed \end{myproof}

 \bibliographystyle{IEEEtran}
\bibliography{../../bib/alias,../../bib/SMD-add,../../bib/JC,../../bib/SM}
\end{document}